\documentclass[10pt,doublecolumn]{IEEEtran}
 \pdfoutput=1
\normalsize
\usepackage{amsthm,amssymb,amsmath,tikz,graphics,cite,float,graphicx,epstopdf,epsfig,verbatim,url}
\usetikzlibrary{automata}
\usetikzlibrary{shapes,arrows}
\newtheorem{lem}{Lemma}
\newtheorem{corol}{Corollary}
\usepackage[utf8]{inputenc}
\newtheorem{thm}{Theorem}

\newtheorem{rem}{Remark}

\newtheorem{mydef}{Definition}

\usepackage{capt-of}
\usepackage[noend]{algpseudocode}
\usepackage{algorithmicx}
\usepackage[ruled]{algorithm}

\tikzset{
  treenode/.style = {align=center, inner sep=0pt, text centered,
    font=\sffamily},
  arn_r/.style = {treenode, circle, black, draw=black, 
    text width=1.5em, very thick},
}
 
\begin{document}
\title{Capacity of Cellular Networks with Femtocache}
\author{Mohsen~Karimzadeh~Kiskani$^{\dag}$,
        and Hamid~R.~Sadjadpour$^{\dag}$,
\thanks{M. K. Kiskani$^{\dag}$ and H. R. Sadjadpour$^{\dag}$ 
are with the Department of Electrical Engineering, University of California, Santa Cruz. Email: 
\{mohsen, hamid\}@soe.ucsc.edu}}
\maketitle \thispagestyle{empty}
\begin{abstract}
The capacity of next generation of cellular networks using femtocaches is studied when multihop communications and decentralized cache placement are considered. 
We show that 
the storage capability of future network User Terminals (UT) can be effectively used to increase the capacity 
in random decentralized uncoded caching. We further propose a random decentralized coded caching scheme 
which achieves higher capacity results than the random decentralized uncoded caching. The result shows that coded 
caching which is suitable for systems with limited storage capabilities can   improve the capacity of 
cellular networks by a factor of $\log (n)$ where $n$ is the number of nodes served by the femtocache.
\end{abstract}


\IEEEpeerreviewmaketitle

\section{Introduction}
Future cellular networks require the support for high data rate video and content delivery.
Many  researchers have recently  focused on 
proposing robust solutions to efficiently address the bandwidth utilization problem. 
For example, the authors in \cite{chandrasekhar2008femtocell} proposed to create 
home sized femtocells to overcome this issue. 

Golrezaei et. al \cite{golrezaei2012femtocaching} proposed an alternate solution by introducing the concept of 
femtocaching. In their solution, several {\em helper} nodes with high storage capabilites 
are deployed in each cell to create a distributed wireless caching infrastructure.
These nodes will reduce the communication burden on the base station by satisfying many of 
the  {\em User Terminal (UT)} requests using the stored contents in their caches. Therefore, the storage capability 
of helper nodes is used to increase the overall network capacity.

Currently, many researchers recommend to utilize  high bandwidth Device-to-Device (D2D)
and Machine-to-Machine (M2M) communication capabilities for UTs. Current IEEE 802.11ad standard \cite{ieee80211ad} and 
the millimeter-wave proposal for future 5G networks 
\cite{boccardi2014five,rappaport2013millimeter}
are examples of such  high bandwidth D2D communications which can enable up to hundreds of GHz of bandwidth.
Authors in \cite{kiskani2015multihop} suggest to use 
this abundant bandwidth to deliver the contents from the helper nodes to the UTs through multihop communications. 
Therefore, they extend the solution in \cite{golrezaei2012femtocaching} to allow   
multihop communication between the helper and the UTs. This approach can significantly reduce network deployment
and maintenance costs without imposing  restrictions on content delivery.

On the other hand, multihop communication between the helper node and the UTs 
together with the use of  UTs' storage capabilities  can improve the overall network capacity. 
Current  improvements on the storage capacity of mobile devices show that future 
UTs will have considerable under-utilized storage capabilities which can be effectively used to improve the network
content delivery. Utilizing the storage capability of UTs allows future cellular networks to move toward a  
distributed D2D wireless caching network without imposing significant communication burden on the base station.

In this paper, we consider a wireless cellular network in which several helper nodes are deployed throughout 
the network to create a wireless distributed caching infrastructure. Each helper is serving a
wireless ad hoc network of UTs through multihop communications as proposed in 
\cite{kiskani2015multihop}. 
We assume that helpers are connected to the base station through a high bandwidth backhaul  infrastructure and  have 
access to all contents. They will use multihop communications to deliver the contents to the UTs. We assume 
that the UTs also use their under-utilized storage capacity to improve network content delivery. 
We will compute the capacity of such networks under decentralized random
coded and uncoded cache placement algorithms. 

In decentralized cache placement algorithms, each UT's cache is populated independently of other UTs. 
In a random decentralized uncoded cache placement algorithms, contents are chosen randomly and stored in UTs cache locations. 
However, in a random decentralized coded cache placement algorithm, each UT stores a combination of multiple contents in its cache.
The UTs will follow this process until 
their caches are fully populated. Coded cache placement is of interest in systems when the storage capacity of each node is 
limited compared to the total number of contents in the network.

This paper computes the capacity of cellular networks with multihop communications using helper  and relay nodes for both uncoded and coded random decentralized cache placement algorithm. Our prior work \cite{kiskani2015multihop} focused on multihop communications with helper nodes but without using the contents stored by the relay nodes. In this paper, the requests can be satisfied either by the helper node or a relay on the path between requesting node and the helper. 
As far as we know, this is the first paper 
to prove that coded caching which is originally motivated by the  lack of sufficient storage capacity in UTs \cite{lee2015index}
can also increase the network capacity.

The rest of the paper is organized as follows. In section \ref{relwork}, the related work is discussed and  
section \ref{netmod} describes the network model considered in this paper. Section \ref{uncoded_sec} focuses on the capacity computation of 
wireless cellular networks operating under a decentralized random 
uncoded cache placement algorithm and  section \ref{coded_sec} reports the capacity  for a 
random coded cache placement algorithm. Simulation results are reported in 
section \ref{sim_sec} and the paper is concluded in 
section \ref{conc_sec}. 

\section{Related Work}
\label{relwork}
The femtocaching network model is proposed in 
\cite{golrezaei2012femtocaching, shanmugam2013femtocaching} and the capacity improvement for single-hop communication is computed. 
In \cite{kiskani2015multihop}, the authors considered a femtocaching D2D network with multihop relaying of information from the helper to the UTs. They proposed a solution based on index coding in which the helper is utilizing the side information in the UTs to create index codes which are to be multicasted to the UTs. This way, they reduce bandwidth utilization by grouping multiple unicast transmissions into  multicast transmission. However, that paper  does not consider the case of coded side information and also it assumes that the relayed message from the helper cannot be changed based on the information in the relaying UTs.

Caching has  been a subject of recent interest to many researchers. The fundamental limits of caching is studied 
in \cite{maddah2014fundamental}. The results in \cite{maddah2014fundamental} has been extended to include 
decentralized  coded caching strategies in \cite{maddah2013decentralized,pedarsani2014online,hachem2014multi,
karamchandani2014hierarchical}.
Other researchers studied the problem of caching in wireless and D2D networks. Among them are the works of 
authors in \cite{ji2014fundamental,ji2013wireless,jeon2015wireless}. 
The authors in \cite{jeon2015wireless} have studied the capacity of wireless 
D2D networks with caching in certain regimes. Our work is essentially different from all of these works in 
the sense that the UT always request the content from helper (femtocache) while in these papers, a wireless ad hoc network is 
considered where UTs' requests can be satisfied by any of the nodes in the network. Clearly, such network model requires significant overhead to locate the nearest UT with the requested content while in our approach, the request always is sent toward the helper.

Coded caching has been 
previously suggested \cite{lee2015index,chen2014fundamental} as an efficient caching technique for devices with small storage capacity. Our results demonstrate that apart from the practical importance of coded caching in small storage 
systems, it can be  useful in increasing the capacity of cached networks. 

\section{Network Model}
\label{netmod}
In this paper we will study the capacity of cellular networks utilizing a distributed femtocaching 
infrastructure as proposed in \cite{golrezaei2012femtocaching}. In these networks, it is assumed that 
several helpers with high storage capacity are deployed throughout the network to assist  
the base station in delivering the contents to the UTs. The UTs can receive contents from helpers using D2D communications 
through either single hop \cite{golrezaei2013femtocaching} or multiple hops \cite{kiskani2015multihop}.

Assume that a helper is serving a D2D network of $n$ nodes. To analyze the capacity of this network, we will use the deterministic
routing 
approach proposed in \cite{kulkarni2004deterministic}. Without loss of generality, it is assumed that the UTs are distributed on a 
square of area 
one and the helper is located at the center of the square as shown in Figure \ref{fig_model}.
\begin{figure}
    \center
      \includegraphics[scale=0.2,angle=0]{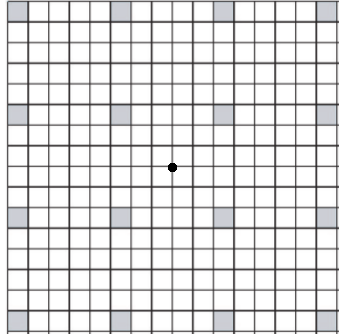}
\caption{The helper node which is located at the center of the unit square is serving $n$
UTs which are  randomly  distributed on a unit square. The square is divided into 
$\Theta(\frac{n}{\log n})$ square-lets of area $\Theta(\frac{\log n}{n})$. 
Gray square-lets can transmit simultaneously. Around each grey square-lets
there is a ``silence'' region of square-lets that are not allowed to transmit
in the given time slot.}
\label{fig_model}
\end{figure}

When a UT  requests a content from 
the helper, the content is routed from the helper to the UT in a sequence of horizontal and vertical
square-lets that are crossing the straight line which 
connects the helper to the UT. 
It is proved in \cite{kulkarni2004deterministic} that if 
the UTs are uniformly distributed over the unit square area and  
the area is divided into $\Theta(\frac{n}{\log n})$ 
square-lets each with area $\Theta(\frac{\log n}{n})$, then with a probability close to one, each square-let contains 
$\Theta(\log n)$ UTs. 
A minimum transmission range of $s(n) = \Theta(\sqrt{\frac{\log n}{n}})$ ensures network 
connectivity \cite{penrose1997longest} in such a dense network. Therefore, assuming a transmission range of $s(n) = \Theta(\sqrt{\frac{\log n}{n}})$, the 
proposed routing algorithm is proved to converge and all the UTs will be able to receive their requested contents with probability one. 

To avoid multiple access interference, a
 {\em Protocol Model} is considered \cite{xue2006scaling} for the successful communication between UTs. 
According to this model, if the UT $i$ is placed at the coordinates $Y_i$, then a transmission from $i$ to another UT $j$ is
successful if $|Y_i-Y_j | < s(n)$ and for any other UT $k$ transmitting on the same frequency band, $|Y_k-Y_j| > (1 + \Delta)s(n)$
for a fixed guard zone factor $\Delta$. A Time Division Multiple Access (TDMA) scheme is assumed for the transmission between the 
square-lets. With the assumption of Protocol Model,
it was shown  \cite{kulkarni2004deterministic} that if the square-lets have a side length of $C_1 s(n)$ for a fixed 
constant $C_1$ and if the square-lets with a distance of $C_2=\frac{2+\Delta}{C_1}$ 
square-lets apart from each other transmit simultaneously,
then there will be no interference between the concurrent transmissions. 

Lets denote the data rate for each UT by $\lambda$, the number of hops between each UT and its helper by $x$, its average
value by $\mathbb{E}[x]$, and the total network throughput by $n \lambda$. 
Therefore, on average the network delivers $n \lambda \mathbb{E}[x]$ bits in a unit of time. There are exactly 
$\frac{1}{(C_2 C_1 s(n))^2}$ square-lets at any time slot available for transmission  and if the total network bandwidth is
$W$ which  is a constant  value independent of $n$, then the total number of bits that the network is capable of delivering is 
upper bounded \cite{azimdoost2013} by $\frac{W}{(C_2 C_1 s(n))^2}$. Hence,
\begin{equation}
\label{capa}
 \lambda \le \lambda_{\max}=\dfrac{W}{n  \mathbb{E}[x] (C_2 C_1 s(n))^2} = \Theta 
 \left(\dfrac{1}{\mathbb{E}[x] \log n} \right).
\end{equation}
This result implies that the maximum throughput can be derived by computing $\mathbb{E}[x]$. The capacity problem 
is therefore reduced to computing the average number of hops traveled between the UTs and the helper.

We assume  the number of contents in the network is $m$ which grows polynomially 
with $n$ \cite{jeon2015wireless} as $m=C_3 n^{\alpha}$. We denote the set of indices of all contents by $\xi=\{1,2,\dots,m\}$. 
Without loss of generality we assume that the contents with lower indices are more popular compared to the ones with higher indices. 
We further assume that the contents can be categorized into two groups of highly popular contents and less popular contents.
Let's denote the requested content by $r$, then the probability that $r$ belongs to the highly popular group of contents should 
be close to one. The highly popular and less popular groups can be defined as 
\begin{mydef}
{\em For  $\epsilon \in (0,1)$, define $h_{\epsilon}$ as the smallest integer such that if $\xi_{1-\epsilon} = \{1,2,\dots,h_{\epsilon}\}$ and 
$\xi_{\epsilon} = \{h_{\epsilon}+1,h_{\epsilon}+2,\dots,m\}$, then  $\textrm{P}[r \in \xi_{1-\epsilon}] \ge 1-\epsilon$.
} 
\end{mydef}
We refer to $\xi_{1-\epsilon}$ as the group of highly popular contents and $\xi_{\epsilon}$ as the group of 
less popular contents. 
We assume that the helper has access to all the contents in $\xi$ but the UTs are assumed to have a limited cache of size
$M=C_4 n^{\beta}$. For the purpose of this paper, we assume that all UTs have the same cache size and the helper (or base station) is applying a decentralized caching strategy  
to populate a UT cache independently of other UTs.  Since the UTs have limited cache size, we assume that only popular contents in $\xi_{1-\epsilon}$ are stored in UTs caches. 
Any request for the less 
popular contents from $\xi_{\epsilon}$ will be satisfied directly by the helper or base station. 

When a UT $i$ requests a content, if that specific content or a
group of coded contents which can be used to decode the content are available in the caches of the UTs in the routing 
path between the UT $i$ and the helper, then the helper informs the UTs which have the coded contents in their caches to 
send the content to UT $i$. If the content or a set of coded contents do not exist 
in the caches of the UTs between UT $i$ and helper, then the content is routed to UT $i$ from the helper through on average 
$s(n)^{-1} = \Theta( \sqrt{\frac{n}{ \log n}})$ hops. Since majority of the requests are from popular contents, these requests can be satisfied by the UTs instead of helper which reduces the average number of transmissions per request. 
Therefore, 
provided that the content request probability distribution is known, the average  number of traveled hops in the network 
can be written as 
 \begin{eqnarray}
 \mathbb{E}[x] &=& \mathbb{E}[x | r \in \xi_{\epsilon}] \textrm{P}[r \in \xi_{\epsilon}] +  \mathbb{E}[x | r \in \xi_{1-\epsilon}] \textrm{P}[r \in \xi_{1-\epsilon}], \nonumber \\
 &=& \epsilon \sqrt{\frac{ n}{\log n}} + (1-\epsilon) \mathbb{E}[x | r \in \xi_{1-\epsilon}].
 \label{ex112}
\end{eqnarray}
\begin{rem}{\em
  By choosing $\epsilon=\frac{1}{\sqrt{n}}$, 
  the average hop count of the contents in $\xi_{\epsilon}$ will become less than one and therefore  
  the total average hop count can be approximated by the average hop counts of the 
  files in $\xi_{1-\frac{1}{\sqrt{n}}}$, i.e., 
   \begin{eqnarray}
 \mathbb{E}[x] &=& \frac{1}{\sqrt{n}} \sqrt{\frac{ n}{\log n}} + 
 (1-\frac{1}{\sqrt{n}}) \mathbb{E}[x | r \in \xi_{1-\frac{1}{\sqrt{n}}}] \nonumber \\
 &=& \Theta( \mathbb{E}[x | r \in \xi_{1-\frac{1}{\sqrt{n}}}] )
 \label{ex1532}
\end{eqnarray}
 }\label{dmfbfj}
\end{rem}
For many web applications \cite{breslau1998implications,breslau1999web}, the  content request popularity follows Zipfian-like distributions. 
Although we will express our results in general form without any specific assumption,
we will later compute explicit capacity results assuming a Zipfian content popularity distribution. 
Our main results in proving the gain of coded caching over uncoded caching
is independent of  the  content popularity distribution. 

For a Zipfian content popularity distribution with parameter $s$, 
the probability of requesting a content with popularity index $i$ will have the 
form 
$ \textrm{P}[r = i] = \frac{i^{-s}}{\sum_{j=1}^m j^{-s}} = \frac{i^{-s}}{H_{m,s}},
$
where $H_{m,s}$ represents the generalized harmonic number with parameter $s$.
\begin{rem}{\em
 In case of Zipfian distribution with $s>1$,
when  few popular contents are widely 
requested by the UTs, we have 
 \begin{eqnarray}
  \textrm{P}[r \in \xi_{\epsilon}] = \sum_{i=h_{\epsilon}+1}^m \frac{i^{-s}}{H_{m,s}} \le
  \frac{(m-h_{\epsilon})(h_{\epsilon})^{-s}}{H_{m,s}}.
  \label{enbg3}
 \end{eqnarray}
 Assuming that $m=C_3 n^{\alpha}$ is a large number, $H_{m,s}$ converges to Reimann Zeta function $\zeta(s)$. Since 
 the number of popular contents is negligible compared to the total number of contents, $\frac{m-h_{\epsilon}}{H_{m,s}}$ can be upper bounded by 
 $\frac{2m}{\zeta(s)}$ and therefore in case of a Zipfian distribution with $s>1$, we have 
  \begin{eqnarray}
  \textrm{P}[r \in \xi_{\epsilon}] \le \frac{2C_3 n^{\alpha}(h_{\epsilon})^{-s}}{\zeta(s)}.
  \label{enbg33dsf}
 \end{eqnarray}
 In order to compute $ h_{\frac{1}{\sqrt{n}}} $ such that  $\textrm{P}[r \in \xi_{\frac{1}{\sqrt{n}}}] \le \frac{1}{\sqrt{n}}$, it is sufficient to have   
 \begin{equation}
  h_{\frac{1}{\sqrt{n}}} = \Theta \left(n^{\frac{1}{s}(\alpha + \frac{1}{2})} \right).
  \label{jdfdh}
 \end{equation}
 \label{derreol1}
 Since we implicitly assume that $h_{\frac{1}{\sqrt{n}}} = O(m) = O (C_3 n^{\alpha})$, equation \eqref{jdfdh} is valid when 
 $\alpha > \frac{1}{2(s-1)}$.
 }
\end{rem}
\begin{rem}{\em
 In case when $\beta > \frac{1}{s}(\alpha + \frac{1}{2})$, 
 the average number of traveled hops can be zero since in that case, all UTs 
 can store all the popular 
 contents in their caches. Therefore, in this case, the maximum per node capacity $\Theta(1)$ is trivially achievable.
 \label{mdgfb}
}\end{rem}
For the rest of paper, we compute capacity assuming that the number of popular contents $h_{\epsilon}$
is known. The capacity for the special case of Zipfian distribution will be derived as well. 

\section{Decentralized Uncoded Caching}
\label{uncoded_sec}
This section focuses on computing the capacity of cellular networks when UTs cache uncoded contents in a distributed fashion. It is assumed that the UTs only cache the most popular contents from $\xi_{1-\epsilon}$.
\begin{lem}
 {\em  If  a  content is drawn uniformly at random from the set of most popular contents in 
 $\xi_{1-\epsilon}$, then the average required number of requests to have at least one copy 
 of each content from $\xi_{1-\epsilon}$ is equal to 
 \begin{equation}
  \mathbb{E}[l] = h_{\epsilon} H_{h_{\epsilon}} = h_{\epsilon} \sum_{i=1}^{h_{\epsilon}} \frac{1}{i}
  = \Theta(h_{\epsilon} \log h_{\epsilon}),
  \label{codhf}
 \end{equation}
 where $H_{h_{\epsilon}}$ is the $h_{\epsilon}^{th}$ harmonic number. 
 This problem is similar to the well-known coupon collector problem. 
 }
 \label{leme0}
\end{lem}
\begin{proof}
 Denote $t_i$ as the number of required requests to collect the $i^{th}$ content after  $(i-1)^{th}$ content have been 
 collected. Notice that the probability of collecting a new content given that $i-1$ contents have been collected is 
 equal to $p_i = \frac{{h_{\epsilon}}-(i-1)}{{h_{\epsilon}}}$. Therefore, $t_i$ has
 geometric distribution with expected value of  
 $\frac{1}{p_i}=\frac{{h_{\epsilon}}}{{h_{\epsilon}}-(i-1)}$.  By the linearity of expectation we have:
 \begin{eqnarray}
  \mathbb{E}[l]=\sum_{i=1}^{{h_{\epsilon}}} \mathbb{E}[t_i] =  
  \sum_{i=1}^{{h_{\epsilon}}} \frac{{h_{\epsilon}}}{{h_{\epsilon}}-(i-1)} 
  = {h_{\epsilon}} \sum_{i=1}^{{h_{\epsilon}}} \frac{1}{i} = h_{\epsilon} H_{h_{\epsilon}} \nonumber
  \label{bsdhgtr}
 \end{eqnarray}
\end{proof} 
\begin{rem}{\em 
 Notice that  
 the contents in UT caches should be stored such that each UT does not 
 cache a content more than once. Therefore, this problem cannot 
 exactly be modeled by the coupon collector problem but  
 when $h_{\epsilon} >> M$, the probability of having 
 multiple instances of the same content in one UT goes 
 to zero and hence the above argument is valid and 
 $\mathbb{E}[l]= h_{\epsilon} H_{h_{\epsilon}}$. 
 }\label{rem_explain}
\end{rem}
Note that during placement phase, we cache contents from the popular set $\xi_{1-\epsilon}$ inside the UTs independently and with uniform distribution. The distribution of  placement of contents inside the UTs is different from the popularity distribution of the contents. 
\begin{thm}{\em 
 In a cellular network with femtocaching operating under a
 decentralized uncoded caching assumption, the average number of traveled hops is  
  \begin{equation}
  \mathbb{E}[x]=\mathbb{E}[x | r \in \xi_{1-\frac{1}{\sqrt{n}}}] = 
  \Theta \left(\frac{h_{\frac{1}{\sqrt{n}}} \log h_{\frac{1}{\sqrt{n}}}}{M} \right).
  \label{ex_uncoded}
 \end{equation}
 Therefore, the following capacity is achievable
    \begin{equation}
  \lambda_{\textrm{uncoded}} = \Theta(\frac{1}{ \mathbb{E}[x] \log n})= 
  \Theta \left(\frac{M}{h_{\frac{1}{\sqrt{n}}} \log h_{\frac{1}{\sqrt{n}}}  \log n} \right).
  \label{capa_uncoded}
 \end{equation}
 }
 \label{thm_uncoded}
\end{thm}
\begin{proof}
 Lemma \ref{leme0} shows that the average number of cache places needed so that all of the requests can 
 be satisfied is $\Theta(h_{\frac{1}{\sqrt{n}}} \log h_{\frac{1}{\sqrt{n}}})$. Since each UT has a cache size 
 $M$, it is obvious that the average number of users needed such that at least one copy of each content is available 
 in their caches is $\Theta(\frac{h_{\frac{1}{\sqrt{n}}} \log h_{\frac{1}{\sqrt{n}}}}{M})$. Hence, along the routing 
 path to the helper, the average number of required hops needed so that the UT can reach its desired content is 
 $\Theta(\frac{h_{\frac{1}{\sqrt{n}}} \log h_{\frac{1}{\sqrt{n}}}}{M})$. This proves the theorem. Equation 
 \eqref{capa} can be used to compute the capacity by replacing $\mathbb{E}[x]$ with the above result.  
\end{proof}
We can use equation \eqref{jdfdh} to simplify the results of theorem 
\ref{thm_uncoded}
to the case of Zipfian content request distribution. 
\begin{corol}{\em
 In a cellular femtocaching network with Zipfian content request distribution with parameter $s>1$ and assuming
 $\alpha > \frac{1}{2(s-1)}$, the following capacity result is achievable.
 \begin{equation}
  \lambda_{\textrm{uncoded}}^{\textrm{Zipf}} = 
  \Theta \left(n^{\beta- \frac{1}{s} (\alpha + \frac{1}{2}) } \frac{1}{(\log n)^2}\right)
  \label{capa_uncoded_zipfian}
 \end{equation}
 \label{mikloplo}
}\end{corol}

\vspace{-0.4in}
\section{Decentralized Coded Caching}
\label{coded_sec}
In this section we will find capacity results assuming that the UTs are caching coded contents from the 
set of popular contents in $\xi_{1-\epsilon}$ independently of 
other UTs. We propose a random coding strategy and we will prove that if UTs follow this random coding 
strategy, the capacity will be increased by a factor of $\log h_{\epsilon}$. The result proves that not only coded 
caching is more efficient in small storage scenarios \cite{lee2015index}, but also it increases the capacity. We first describe  the random coding cache placement and the decoding procedure. 

{\bf Coded cache placement:} For the purposes of this paper, we assume that random coding is done over
Galois Field GF(2). For each encoded file, 
the helper node (or base station) randomly selects each one of contents from the set $\xi_{1-\epsilon}$
with probability $\frac{1}{2}$ and then add all the selected contents  to create one encoded file. For a UT with 
cache size $M$, the helper node creates $M$ of these encoded files. Therefore, each one of the contents in $\xi_{1-\epsilon}$ has 
been used on average $\frac{M}{2}$ times to create the $M$ coded files. 

{\bf Coded file reconstruction:} When a UT requests a content, if the content is among the set of popular contents 
$\xi_{1-\epsilon}$, it sends the request to the helper. The helper then decides to send the file through a routing
path as proposed in \cite{kulkarni2004deterministic}. However, it is highly possible that the content can be reconstructed 
using the coded contents in the caches of UTs between the requesting UT and the helper along the routing path. If that 
is the case, then the helper sends appropriate coding information to the relaying UTs along the routing path  and each
relay UT that has useful information, add that information to the file that is being relayed to the requesting UT. This 
procedure continues hop by hop until the content reaches the requesting UT.  After the requesting UT receives this file, 
 it can reconstruct the desired  content by applying its own coding gains to the received coded file. By doing so, there will
 not be multiple transmissions by relaying UTs to construct the requested content. 

%
%
To prove our results 
we will first prove the following lemma. 
\begin{lem}
 {\em 
 If for a vector $v_i \in \mathbb{F}_2^{{h_{\epsilon}}}$, every element is equal to 1 with probability
 $\frac{1}{2}$ and equal to 0 with 
 probability $\frac{1}{2}$ and  $\{ v_1,v_2,\dots,v_q \}$ span the vector  space of 
 $\mathbb{F}_2^{h_{\epsilon}}$, then 
 the average required number of such vectors to span the set $\mathbb{F}_2^{{h_{\epsilon}}}$ equals  
$  \mathbb{E}[q] = {h_{\epsilon}} + \sum_{i=1}^{h_{\epsilon}}
  \frac{1}{2^{i}-1}  =   {h_{\epsilon}} + \gamma $ where 
 $\gamma  \approx 1.6067$ is called the Erdős–Borwein constant.
 }
 \label{leme1}
\end{lem}
\vspace{-0.2cm}
\begin{proof} 
  We can form a Markov chain to model the problem. 
 The states of this Markov chain are equal to the dimension of the space spanned by vectors $v_1,v_2,\dots,v_l$. 
 Let $k_l$ ($k_l \leq h_{\epsilon}$) denote the dimension of the space spanned by vectors $v_1,v_2,\dots,v_l$.
 Therefore, the Markov chain will have $k_l+1$ distinct states. Assuming that we are in state $k_l$, 
 we want to find the 
 probability that adding a new vector will change the state to $k_l+1$. When we are in state $k_l$, adding 
 $2^{k_l}$ vectors out of the total $2^{h_{\epsilon}}$ possible vectors will not change the dimension while adding any one of 
 $2^{h_{\epsilon}} - 2^{k_l}$ new vectors will change the 
 dimension to $k_l+1$. Therefore, the Markov chain can be represented as the one in Figure \ref{markovchain}. 
 \begin{figure}[H]
\begin{center}
\begin{tikzpicture}[->,>=stealth',auto,semithick,node distance=1.8cm]
\tikzstyle{every state}=[fill=white,draw=black,thick,text=black,scale=1]
\node[state]    (k0)                {};
\node[state]    (k1)[right of=k0]   {};
\node[state]    (k2)[right of=k1]   {};
\node[state]    (k3)[right of=k2]   {};
\node[state]    (kn)[right of=k3]   {};
\node [label={[label distance=0.5cm]$k_l=0$}] (t0)[below of=k0]{};
\node [label={[label distance=0.5cm]$k_l=1$}]   (t1)[below of=k1]   {};
\node [label={[label distance=0.5cm]$k_l=2$}]   (t2)[below of=k2]   {};
\node [label={[label distance=0.5cm]$k_l=3$}]   (t3)[below of=k3]   {};
\node [label={[label distance=0.5cm]$k_l={h_{\epsilon}}$}]   (tn)[below of=kn]   {};
\path (k0)    edge[loop above]    node{$\frac{1}{2^{h_{\epsilon}}}$}      (k0);
\path (k0)    edge[above]    node{$1-\frac{1}{2^{h_{\epsilon}}}$}    (k1);
\path (k1)    edge[loop above]    node{$\frac{2}{2^{h_{\epsilon}}}$}      (k1);
\path (k1)    edge[above]    node{$1-\frac{2}{2^{h_{\epsilon}}}$}    (k2);
\path (k2)    edge[loop above]    node{$\frac{2^2}{2^{h_{\epsilon}}}$}    (k2);
\path (k2)    edge[above]    node{$1-\frac{2^2}{2^{h_{\epsilon}}}$}  (k3);
\path (k3)    edge[loop above]    node{$\frac{2^3}{2^{h_{\epsilon}}}$}    (k3);
\path (k3)    edge[dashed]        node{}                     (kn);
\path (kn)    edge[loop above]    node{$1$}                  (kn);
\end{tikzpicture}
\end{center}
\vspace{-1cm}
\caption{The state space of the Markov chain used in proof of lemma \ref{leme1}.}
\label{markovchain}
\end{figure}
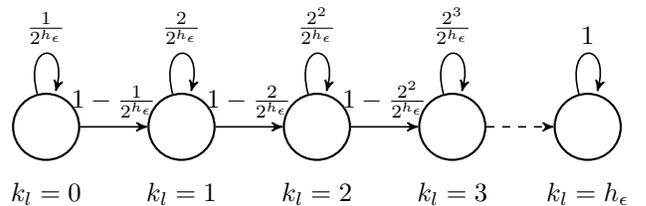 
 The state transition matrix for this Markov chain 
can be written in the form of a discrete phase-type distribution as
\begin{equation} 
P=\begin{bmatrix} 
T & T_0  \\ 
0 & 1 \\ \end{bmatrix},
\label{msdbf}
\end{equation}
where
\begin{equation} 
T=\begin{bmatrix} 
\frac{1}{2^{h_{\epsilon}}} & 1- \frac{1}{2^{h_{\epsilon}}} & 0                & \cdots & 0         &0               \\ 
0             & \frac{2}{2^{h_{\epsilon}}}    & 1- \frac{2}{2^{h_{\epsilon}}} & \cdots & 0         &0                \\ 
0             &   0              & \frac{4}{2^{h_{\epsilon}}}    & \cdots & 0         &0                 \\ 
\vdots        & \vdots           & \vdots           & \ddots & \vdots    &\vdots   \\
0             & 0                & 0                & \cdots           & \frac{2^{{h_{\epsilon}}-2}}{2^{h_{\epsilon}}} & 1- \frac{2^{{h_{\epsilon}}-2}}{2^{h_{\epsilon}}} \\ 
0             & 0                & 0                & \cdots           & 0         &\frac{2^{{h_{\epsilon}}-1}}{2^{h_{\epsilon}}}  \\
\end{bmatrix},
\end{equation}
\begin{equation}
T_0^t= [ 0 \hspace{0.1in} 0 \ldots 1-\frac{2^{{h_{\epsilon}}-1}}{2^{h_{\epsilon}}}],
\end{equation}
and $t$ denotes transpose operation.
If $e$ denotes  all one vector of size ${h_{\epsilon}}$, 
since $P$ is a probability distribution we  have 
$P\begin{bmatrix}e \\ 1 \end{bmatrix}=\begin{bmatrix}e \\ 1 \end{bmatrix}.
$ This 
implies that $Te + T_0 = e$, hence $ T_0 =  (I - T)e$. Therefore, it is easy to show by induction that the state transition matrix in $l$ steps can 
be written as
\begin{equation} 
P^l=\begin{bmatrix} 
T^l & (I-T^l)e  \\ 
0 & 1 \\ \end{bmatrix}.
\label{mskit}
\end{equation}
This equation implies that if we define the absorption  time as
\begin{equation}
q = \inf \{l \ge 1 ~|~ k_l={h_{\epsilon}}\},
\end{equation}
and if $l$ is strictly less than the absorption time, the probability of transitioning from state $i$ to state $j$ by having $l$ new vectors 
can be found from the submatrix $T^l$ of $P^l$. In other words, 
\begin{equation}
 \mathrm{P}_i^l[k_l = j,  l < q] = (T^l)_{ij}.
 \label{probres}
\end{equation}
Therefore, starting from state $i$, if $t_j^i$ denotes the time spent in state $j$ before absorption, 
$t_j^i$ can be written as 
\begin{equation}
t_j^i = \sum_{l=0}^{q-1} \mathrm{1}\{k_l=j\}
\label{msdhlsdjfusdg} 
\end{equation}
Therefore, starting from state $i$,  the average time spent in state $j$ will be equal to 
\begin{eqnarray}
 \mathbb{E}[t_j^i] = \mathbb{E} \left[\sum_{l=0}^{q-1} \mathrm{1}\{k_l=j\} \right] = \sum_{l=0}^{q-1} \mathbb{E} \left[ \mathrm{1}\{k_l=j\} \right]. \nonumber 
 \label{expected1}
\end{eqnarray}
Since $\mathbb{E} \left[ \mathrm{1}\{k_l=j\} \right] = \mathrm{P}_i^l(k_l = j, l \le q-1) $, we have 
\begin{equation}
\mathbb{E}[t_j^i] = \sum_{l=0}^{q-1} \mathrm{P}_i^l(k_l = j, l \le q-1)  
 = \sum_{l=0}^{\infty} \mathrm{P}_i^l(k_l = j, l \le q-1)  \nonumber
\end{equation}
\begin{equation}
 = \sum_{l=0}^{\infty} \mathrm{P}_i^l(k_l = j, l < q)  = \sum_{l=0}^{\infty} (T^l)_{ij}.
 \label{expected2}
\end{equation}
Since the probability is nonzero up to $q-1$, then we can extend the summation to infinity adding zero terms in  \eqref{expected2}. 
Notice that the equality in the last line comes from equation \eqref{probres}.
If we denote matrix $U=(\mathbb{E}[t_j^i])_{ij}$, using equation \eqref{expected1} and using matrix algebra, we have 
\begin{equation}
U = \sum_{i=0}^{\infty}T^i =(I-T)^{-1}.
\end{equation}
It is not difficult to verify that 
\begin{equation}
\nonumber \\
 U=(I-T)^{-1}=
 \begin{bmatrix}
 \frac{2^{h_{\epsilon}}}{2^{h_{\epsilon}}-1} & \frac{2^{{h_{\epsilon}}-1}}{2^{{h_{\epsilon}}-1}-1} &  \frac{2^{{h_{\epsilon}}-2}}{2^{{h_{\epsilon}}-2}-1}& \cdots &2 \\ 
 0 & \frac{2^{{h_{\epsilon}}-1}}{2^{{h_{\epsilon}}-1}-1} & \frac{2^{{h_{\epsilon}}-2}}{2^{{h_{\epsilon}}-2}-1}& \cdots&2\\ 
 0 & 0 & \frac{2^{{h_{\epsilon}}-2}}{2^{{h_{\epsilon}}-2}-1}&  \cdots&2\\ 
 \vdots & \vdots & \vdots & \ddots & \vdots \\ 
 0  & 0 & 0 & \cdots &  2
\end{bmatrix}
\end{equation}
Therefore, starting at $k_l=0$, the average time it takes to get to absorption is equal to 
\begin{eqnarray} 
\mathbb{E}[q] &=& 
\begin{bmatrix}1 & 0 & \cdots & 0\end{bmatrix} U e \nonumber \\
&=& \begin{bmatrix}1 & 0 & \cdots & 0\end{bmatrix} (I-T)^{-1} e \nonumber \\
&=& \sum_{i=1}^{h_{\epsilon}} \frac{2^i}{2^i-1} ={h_{\epsilon}} + \sum_{i=1}^{h_{\epsilon}} \frac{1}{2^i-1} 
\end{eqnarray} 
This proves the lemma. 
 \label{prof1}
\end{proof}
This lemma shows that each UT's request can be satisfied in a smaller number of hops compared to an uncoded caching
strategy. Therefore, the capacity will be increased. The following theorem formalizes this. 
\begin{thm}{\em 
 In a cellular network with femtocaching, 
 our proposed decentralized coded caching in which each popular content in $\xi_{1-\epsilon}$ is 
 present in any cache location with probability $\frac{1}{2}$ reduces the required number of 
 traveled hops for each request by UTs to at most 
  \begin{equation}
  \mathbb{E}[x]=\mathbb{E}[x | r \in \xi_{1-\frac{1}{\sqrt{n}}}] = 
  \Theta \left( \frac{h_{\frac{1}{\sqrt{n}}}}{M}\right).
  \label{ex_coded}
 \end{equation}
 Therefore, the following capacity is achievable through decentralized coded caching.
    \begin{equation}
  \lambda_{\textrm{coded}} = \Theta(\frac{1}{ \mathbb{E}[x] \log n})= 
  \Theta \left(\frac{M}{h_{\frac{1}{\sqrt{n}}}  \log n}\right)
  \label{capa_coded}
 \end{equation}
 }
 \label{thm_coded}
\end{thm}
\begin{proof}
Lemma \ref{leme1} shows that to be able to decode a requested content, on average $\Theta( h_{\frac{1}{\sqrt{n}}})$
coded contents are required. Since each UT has a cache of size $M$, we need $\Theta( \frac{h_{\frac{1}{\sqrt{n}}}}{M})$  UTs to be able to 
decode the desired content. This means that along the routing path, we only need to travel 
 a distance of $\Theta (\frac{ h_{\frac{1}{\sqrt{n}}}}{M})$ hops away from each UT to find 
all the contents that the UT requires for decoding its desired content. 
Notice that  individual UTs do not need to separately send their coded content to the requesting node. Each UT can combine the appropriate encoded files to the received file along the  
route  to the requesting node. 
\end{proof}
Similarly, the results of theorem \ref{thm_coded} can be simplified by using 
equation \eqref{jdfdh} for the case of Zipfian content request distribution. 
\begin{corol}{\em
 In a cellular femtocaching network with Zipfian content request distribution with parameter $s>1$ and assuming
 $\alpha > \frac{1}{2(s-1)}$, the following capacity result 
 is achievable through decentralized coded caching.
 \begin{equation}
  \lambda_{\textrm{coded}}^{\textrm{Zipf}} = \Theta \left(n^{\beta- \frac{1}{s} (\alpha + \frac{1}{2}  )} 
  \frac{1}{\log n}\right)
  \label{capa_coded_zipfian}
 \end{equation}
 \label{miklo}
 \vspace{-0.4cm}
}\end{corol}
{ Our proposed coded caching strategy can be done with 
insignificant overhead as the coding instructions sent from the helper 
is negligible compared to the size of the files. The computational 
complexity of in each UT (XOR operation) is also not significant. However, the 
helper requires to have high computational complexity capability. Future works should
concentrate on the ways to reduce the complexity and delay for this 
approach.}
\vspace{-0.5cm}
\section{Simulations}
\label{sim_sec}
The simulation results  compare the performance of our proposed  decentralized
random coded caching with decentralized random uncoded caching. We assume a helper which is serving $n=2500$ UTs. The
Zipfian content request probability  parameter is $s=2.5$, $\alpha = 1.5$, and $C_3 =8$ which means that a 
total of $m=1000000$ contents exist in the network and 523 popular contents are considered for this simulation.   The cache size parameter 
$\beta$ is ranging from $0.3$ to $0.8$ while $C_4=1$. Figure \ref{fig_sim22} shows the simulation results comparing the average number of hops 
required to decode the content in both decentralized coded and uncoded caching. As can be seen from this figure, our proposed 
decentralized random coded cache placement algorithm can significantly reduce the average number of traveled hops compared to 
decentralized uncoded cache placement. 
Further, the theoretical results match the simulation results for both cases. 

\begin{figure}
    \center
      \includegraphics[scale=0.4,angle=0]{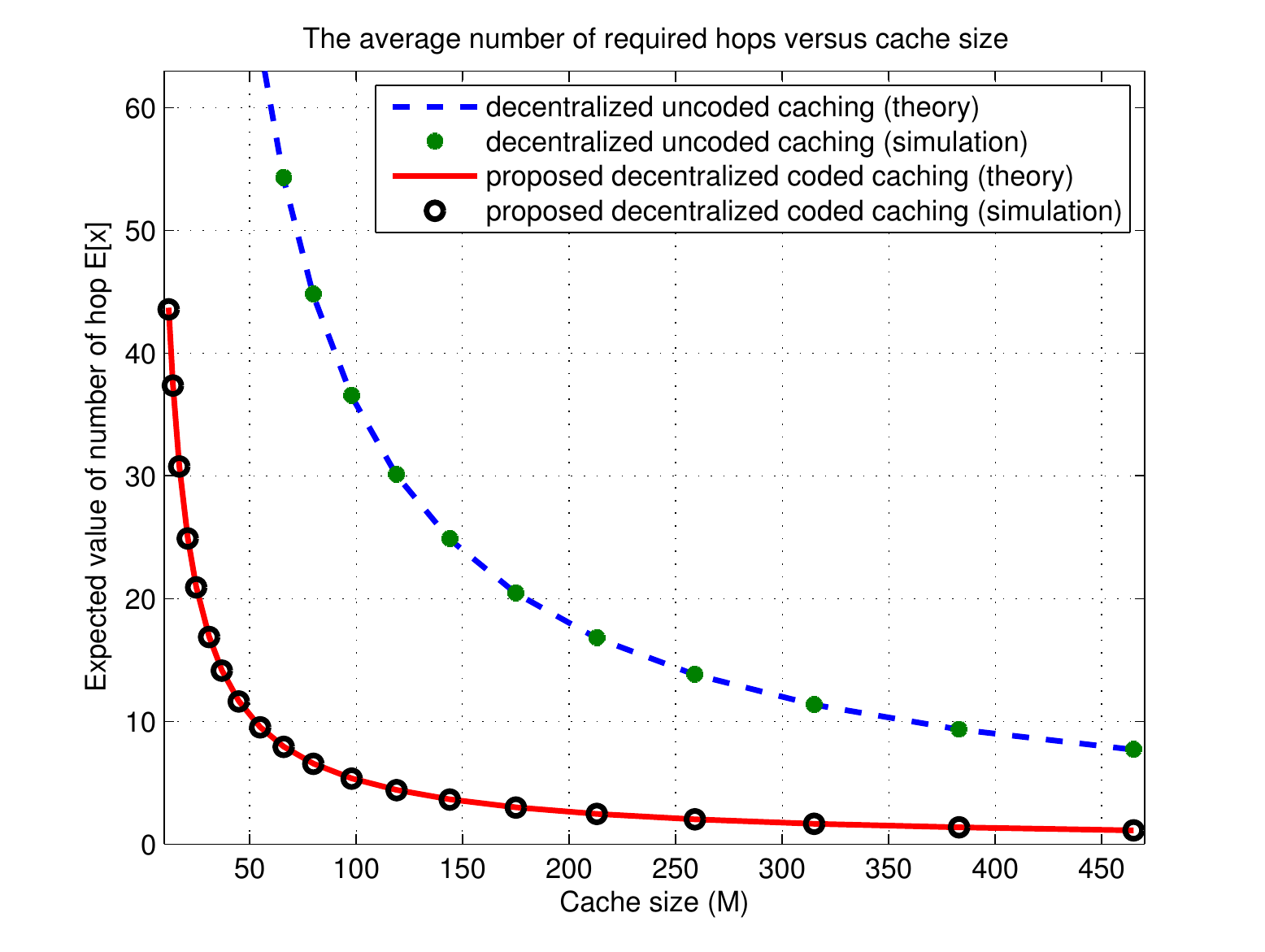}
\vspace{-0.1in}
\caption{Simulation results show that the expected number of hops in case of decentralized coded caching 
is significantly lower than the expected number of hops for decentralized uncoded caching.}
\label{fig_sim22}
\end{figure}
\vspace{-0.3cm}
\section{Conclusions}
\label{conc_sec}
In this paper, we studied the content delivery capacity in cellular networks with femtocaching  
with decentralized uncoded and coded caching for UTs. We computed the capacity of random
decentralized uncoded caching. We then proposed a random 
coded caching strategy for network users and  proved that this random coded caching technique can improve the capacity. 
Note that we did not consider the possibility of congestion near helper node since all contents are moving toward that node. In the future work, we intend to study the effects of congestion on the capacity of the network. 
\bibliographystyle{plain}
\bibliography{All-Papers-infocom}
\end{document}